\documentclass{article}[14pt]

\usepackage{fullpage}
\usepackage{appendix}
\usepackage{amsthm,amsmath,amsfonts,amssymb,amsxtra,bookmark,dsfont,bm,xcolor}
\numberwithin{equation}{section}

\usepackage{mathtools}
\usepackage{enumerate}
\usepackage{comment}
\allowdisplaybreaks
\usepackage[auth-lg]{authblk}

\usepackage{dirtytalk}

\numberwithin{equation}{section}

\theoremstyle{plain}
\newtheorem{theorem}{Theorem}[section]
\newtheorem{proposition}[theorem]{Proposition}
\newtheorem{lemma}[theorem]{Lemma}

\theoremstyle{remark}

\def\f{\frac}

\def\n{\noindent}
\def\vs{\vspace{0.5cm}}
\def\be{\begin{equation}}
\def\ee{\end{equation}}

\def\eps{\varepsilon}

\def\1{{\ensuremath {\mathds 1} }}

\def\R {\mathbb{R}}
\def\N {\mathbb{N}}
\def\C {\mathbb{C}}

\newcommand{\tc}{T_{\rm{coll}}}

\title{A mathematical model for the Einstein-Podolsky-Rosen argument}

\author[1]{Riccardo Adami}
\author[2]{Luigi Barletti}
\author[3]{Alessandro Teta}
\affil[1]{Dipartimento di Scienze Matematiche ``G.L. Lagrange'', Politecnico di Torino, \newline
Corso Duca degli Abruzzi 24, 10129 Torino, Italy\newline }
\affil[2]{Dipartimento di Matematica e Informatica ``U. Dini'', Università di Firenze, \newline Viale G.B. Morgagni 67/a, 50134 Firenze, Italy \newline}
\affil[3]{Dipartimento di Matematica ``G. Castelnuovo'', Università di Roma La Sapienza, \newline P.le A. Moro 2, 00185 Roma, Italy\newline}

\date{\today}

\begin{document}

\maketitle

\begin{abstract}
We study a nonrelativistic system made of two quantum particles  constrained to move on a line and a spin located at a fixed point of the line. 
Initially the two particles are in a maximally entangled state and the spin is down. The first particle interacts with the spin while the second particle is free, i.e., it does not interact neither with the first particle nor with the spin. We rigorously prove that there is a correlation between the state of the spin and the state of the second particle. More precisely, we show that, in a suitable scaling limit, if the first particle flips the spin, then the second  particle possesses a definite  momentum in the direction opposite to the spin.
\end{abstract}

\large

\section{Introduction} \label{sec:intro}

\vs
The celebrated paper by Einstein, Podolsky and Rosen (EPR) published in 1935 \cite{EPR}  is surely one of the most influential works on Quantum Mechanics, not only for its deep conceptual meaning,  but also because it has  highlighted peculiar aspects of the theory that during the  years have led to fundamental applications in many fields, such as quantum information theory.

\n
The aim of this paper is to  formulate a simplified but mathematically well-defined version of the model introduced by EPR and to rigorously derive its physical consequences. 

\n
The original EPR work aims to show that Quantum Mechanics is not complete and it is based on the analysis of  a  model with continuous variables describing two particles in dimension one. 
As is well known, the argument was then reformulated by Bohm \cite{BOHM} in terms of spin variables, which is simpler from a mathematical point of view. Such reformulation is the one commonly discussed and analysed in the literature, while the original EPR formulation is  less  known. However, since we are interested in the latter case, for the convenience of the reader  we  summarise the original  argument in Appendix A.

\n
The model we consider consists of a pair of quantum particles, referred to as particle $1$ and particle $2$, and a spin. The Hilbert space of the  states of the system is then
$$\mathcal K=  L^2 (\R) \otimes L^2 (\R) \otimes \C^2\,, $$
so a state of the system at time $t$ is represented by a two-component vector function, namely
$$ \Psi_t (x_1, x_2) \ = \ \left( \begin{array}{c}\Psi_t^u (x_1, x_2) \\ \Psi_t^d (x_1, x_2) \end{array}\right),$$
where $\Psi_t^{u}, \Psi_t^{d} \in L^2(\R) \otimes L^2(\R)$ and the variable $x_j$, $j=1,2$,  denotes the position variable of the $j-$th particle of the system. The upper component $\Psi_t^u (x_1, x_2) $
is associated with the value $\hbar / 2$ of the spin, i.e., with the spin up,  while the lower $ \Psi_t^d (x_1, x_2) $
refers to the value $-\hbar / 2$, i.e., with the spin down. 

\noindent
Particle $2$ is free, while particle $1$ can interact with the spin through a localized potential $\gamma V$, where $\gamma \in \R$ is a coupling constant and $V$ is smooth and rapidly decaying at infinity. 
Furthermore, the interaction can make the spin flip, so that the full interaction reads as $\gamma V {\bf \sigma}_2$, where
$$ {\bf \sigma}_2 = \left( \begin{array}{cc}
    0  & -i  \\
    i &  0
\end{array}\right)
$$
is the second Pauli matrix. The Hamiltonian of the system reads then
\begin{equation} \begin{split}
    \label{h}
    {\mathcal H} \ = & \ {\mathcal H}_0 + 
     \gamma \, V \Big( \frac{x_1 - a} \delta \Big) \, {\bf \sigma}_2 \\
    {\mathcal H}_0 \ = & \  - \frac{\hbar^2}{2m_1 } {\mathbb I} \, \Delta_1 
    - \frac{\hbar^2}{2m_2 } {\mathbb I} \, \Delta_2 + \frac{\hbar \omega}{2}{\bf \sigma}_3 
\end{split}
\end{equation}
where:\\
$m_1$ and $m_2$ are the masses of the particles;\\
 $\omega$ is the characteristic frequency of the spin and therefore the free evolution of the spin is periodic with period $T_s = 2 \pi / \omega$;\\
${\mathbb I}$ is the $2 \times 2$ identity matrix;\\
${\bf \sigma}_3$ is the Pauli matrix
$$ {\bf \sigma}_3 = \left( \begin{array}{cc}
    1  & 0  \\
    0 &  -1
\end{array}\right);
$$
$\gamma$ is the coupling constant of the potential $V$ that describes the interaction between particle $1$ and the spin;\\
 $\delta > 0$ is the spatial scale of the range of $V$;\\
 $a>0$ is the position of the spin, that is fixed.

\noindent
We study the time evolution of the system 
under the assumption that the initial datum has the form
\begin{equation}\label{id}
 \Psi_0 \ = \ 
 \left(\!\! \begin{array}{c} 0 \\ \Psi_0^d  \end{array} \! \! \right),
 \end{equation} 
so that the spin equals $- \hbar / 2$ and the associated energy is $- \hbar \omega / 2$.
Moreover, the pair of particles is in a maximally entangled state concentrated in position around the origin, namely
\begin{equation} \label{psizzeromeno}
    \Psi_0^d (x_1, x_2) \ = \ N ( \phi_P (x_1) \phi_{-P} (x_2) +   \phi_{-P} (x_1) \phi_{P} (x_2) ),
    \end{equation}
where 
$$ 
\phi_P(x)= \frac{1}{\sqrt{\sigma}} f\Big(\frac{x}{\sigma}\Big) \, e^{i \frac{P}{\hbar} x}\,, \;\;\;\;\;\; \sigma>0\,, \;\; P>0\,. 
$$
Here $f$ is smooth, rapidly decaying,  with $f$ even, real and  $\|f\|_{L^2(\R)} =1$. 
The normalization factor $N$ is given by
$$
N= \left( 2 + 2 \int\!\! dy \, f(y)^2 \cos \frac{2 \sigma P}{\hbar} y \right)^{\!\! - 1/2}
$$
We stress that \eqref{psizzeromeno} is a superposition of the state $\phi_P (x_1) \phi_{-P} (x_2)$, where  particle $1$ has mean momentum $P$ and particle $2$  mean momentum $-P$, and the state $\phi_{-P} (x_1) \phi_{P} (x_2)$, where mean momenta are exchanged.

\noindent
In order to let the EPR phenomenon emerge from the dynamics we need to characterize the parameters of the model. More precisely, 
the physical situation we want to describe  is the following. 
Let us define  
\be
T_{coll} := a\, \f{m_1}{P}\,
\ee
i.e.,  the classical collision time of  particle $1$ with momentum $P$ starting from the origin with the spin placed in $a$.  

\n
For  $t < T_{coll}$ we require that, with high probability, the system is described by the free evolution. 
This means that the spin remains down and the particles are described by a superposition  given by the free evolution of  $\phi_P (x_1) \phi_{-P} (x_2)$ and the free evolution  of  $    \phi_{-P} (x_1) \phi_{P} (x_2)$. 

\n
In particular this means that for $t< T_{coll}$ each particle has  positive momentum with probability approximately $1/2$ and negative momentum with probability approximately $1/2$.

\n
 At $t \simeq T_{coll}$ particle $1$   interacts with the spin which, with some probability, can flip to state up. More precisely,  the interaction occurs for the particles described by  the time evolution of the component $\phi_P(x_1) \phi_{-P} (x_2)$ of \eqref{psizzeromeno}. 
 Recall that particle $2$ does not interact with the spin and, moreover, in the free evolution of the component  $\phi_{-P} (x_1) \phi_{P} (x_2)$ particle $1$ remains far from the position of the spin.

 \n
For $t> T_{coll}$ we expect that the system is described by a superposition representing the two possible situations (apart for some small errors): 

\n
$(i)$ the spin is down and the two particles are still described by the free evolution of the initial state \eqref{psizzeromeno},  

\n
$(ii)$ the spin is up and the two particles are described by a product state, with particle $1$ evolving freely with (approximate) mean momentum $P $ and mean position on the right of the spin position $a$, while particle $2$ evolves freely with mean momentum $-P$ and mean position on the left of the origin.

\n
Notice that situation $(ii)$ is the one of interest for the EPR argument, in the sense that if the spin is up then particle $1$ has mean momentum $P$ and therefore particle $2$ has mean momentum $-P$. We shall come back to this point with more details after the precise formulation of our result.

\n
The physical situation described above cannot be realized for arbitrary values of the parameters present in  the model. It should be clear that we need a  semiclassical regime for the two particles,  a small and short range interaction potential and  excitation energy of the spin much smaller than the initial kinetic energy of particle $1$, i.e., a quasi-elastic regime for particle $1$. In order to impose these conditions, it is convenient to introduce the following scaling   (analogous to the one used in  \cite{FigariTeta,RT}):
\begin{equation}
    \label{scaling}
    \hbar = \varepsilon^2, \ \;\;\; \omega = \eps^{-1}, \ \;\;\; \gamma = \eps^2,  \ \;\;\; \delta = \eps \ \;\;\; \sigma = \eps
\end{equation}
where $\eps$ is a small positive parameter, while $P$ and $a$ are of order 1. 
For notational simplicity, we also set $m_1 = m_2=1$.

\n
Let us briefly comment on the  meaning of our scaling for $\eps \to 0$. 
We first  note that  
the quantity
$
\frac{\Delta p}{P}$ is roughly given by $ \frac{\hbar}{P\sigma} =O(\eps),
$
which means that the momentum of the particles is well concentrated around the mean values $P$ or $-P$. Moreover, 
$$
\frac{\sigma}{a} =O(\eps)\,, \;\;\;\;\;\; \frac{\delta}{a} =O(\eps)
$$
i.e. the localization in position of the particles and the effective range of the interaction are much smaller than the distance $a$. We also stress that the spacing of the energy levels of the spin  is $\hbar \omega = O(\eps)$ while the kinetic energy of particle $1$ is $O(1)$ (quasi-elastic regime for particle $1$). 
Moreover, the choice of the coupling constant  $\gamma =\eps^2$ guarantees applicability of perturbative methods.

\n
Furthermore, it is interesting to compare the characteristic times of our system. We have

$$
\tc  = \frac{a}{P} =O(1) 
$$
while the period of the \say{spin motion} is 
$$
T_s= \frac{2\pi}{\omega} =O(\eps)
$$
and the time necessary to particle $1$ to cross the interaction region is 
$$
T_{int}= \frac{\delta}{P} =O(\eps)\,.
$$
Therefore the condition $\tc \gg T_s, T_{int}$ allows us to interpret $\tc$ as the effective collision time in the quantum description.  Moreover, the condition $T_s / T_{int} = O(1)$   proves crucial to have a non trivial transition probability for the spin.

\n
According to the scaling \eqref{scaling}, the Hamiltonian of the system reads

\begin{equation} \begin{split}
    \label{he}
    {\mathcal H}^{\eps} \ = & \ {\mathcal H}_0^{\eps} + 
     \eps^2 \, V^{\eps} \, {\bf \sigma}_2 \\
    {\mathcal H}_0^{\eps}  \ = & \  - \frac{\eps^4}{2 } {\mathbb I} \, \Delta_1 
    - \frac{\eps^4}{2 } {\mathbb I} \, \Delta_2 + \frac{\eps}{2}{\bf \sigma}_3 
\end{split}
\end{equation}
where $V^{\eps}$ denotes the multiplication operator by $V\big(\frac{x_1 - a}{\eps} \big)$. 
The initial state has  the form \eqref{id}, with $\Psi_0^-$ replaced by

\begin{equation} \label{psizzeromenoe}
    \Psi_0^{-,\eps} (x_1, x_2) \ = \ N^{\eps} ( \phi_P^{\eps} (x_1) \phi_{-P}^{\eps} (x_2) +   \phi_{-P}^{\eps} (x_1) \phi_{P}^{\eps} (x_2) ),
    \end{equation}
where 
$$ 
\phi_P^{\eps}(x)= \frac{1}{\sqrt{\eps}} f\Big(\frac{x}{\eps}\Big) \, e^{i \frac{P}{\eps^2} x}
$$
and $N^{\eps}$ is given by
$$
N^{\eps} = \left( 2 + 2 \int\!\! dy \, f(y)^2 \cos \frac{2  P}{\eps} y \right)^{\!\! - 1/2}\!\!.
$$
Notice that $N^\eps = 1/\sqrt{2} + O(\eps^n)$ for all $n$.

\n
With an abuse of notation, from now on we  shall drop the dependence on $\eps$ of the initial state, the Hamiltonians and the corresponding unitary propagators.

\n
In order to formulate the result we define the interacting propagator
\begin{equation}
\mathcal U (t) = e^ {-i \frac{t}{\eps^2} \mathcal H}
\end{equation}
 and the free propagator
\begin{equation}
\mathcal U_0 (t) = e^ {-i \frac{t}{\eps^2} \mathcal H_0}
\end{equation}
 both acting in $\mathcal K$. Note that  the action of $\mathcal U_0(t)$ is explicitly given by 
 
\begin{align} \label{propagatore}
& \mathcal U_0(t) \Psi =  \left(\!\! \begin{array}{c} e^{- i \frac{t}{2 \eps} } U_0(t) \otimes U_0(t) \Psi^u  
\\ e^{ i \frac{t}{2 \eps} } \,\,  U_0(t) \, \otimes U_0(t) \Psi^d \end{array} \! \! \right) \;\;\;\;\;\; \text{for any }\;\;\;\; 
\Psi=  \left(\!\! \begin{array}{c} \Psi^u  
\\  \Psi^d \end{array} \! \! \right) \in \mathcal K\,.
\end{align} 
where $U_0(t)$ is  the free propagator in $L^2(\R)$
\begin{equation}\label{freepro}
  U_0 (t) f (x) \ = \ \frac{1}{\sqrt{2 \pi i \eps^2 t}} \int_\R
e^{i  \frac{ (x - y)^2 }{2 \eps^2 t}} f (y) \, dy\,. 
\end{equation}

 \n
Our main result is summarized in the following theorem.

\begin{theorem}\label{th11}
Assume that $V$ and  $f$ belong to the  Schwartz space $\mathcal S(\R)$, with $f$ real, even and $\|f\|_{L^2(\R)} =1$, and let us fix $t> T_{coll}$. Then for $\eps \to 0$ we have

\begin{equation}
\Psi_t:=  \mathcal U(t) \Psi_0  =   \mathcal U_0(t) \Psi_0  + \eps \,  \mathcal U_0(t) \Psi^{(I)}  +\mathcal R(t)
\end{equation}
where
\begin{align}
& \Psi^{(I)}(x_1,x_2) =  N \left(\!\! \begin{array}{c} A(x_1) \, \phi_{-P}(x_2)  \\ 0  \end{array} \! \! \right), \label{psii}\\
&A(x) =  -  \frac{\sqrt{2 \pi}}{P} \, e^{i \frac{a}{\eps P} \left( 1 + \frac{\eps}{2 P^2} \right) } \,\overline{ \hat{V}(P^{-1}) } \,\frac{1}{\sqrt{\eps}}  \,
f \!\left(\frac{x - a P^{-2} \eps}{\eps} \right)  e^{ \frac{i}{\eps^2} \left( P - \eps P^{-1} \right) x} \,, \label{th2}\\
&\|\mathcal R(t)\|_{\mathcal K}< C \, \eps^2\,\label{th3}
\end{align}
and the positive constant $C$ is independent of $\eps$ and depends on $t$ and on the parameters of the model.

\end{theorem} 

\n
The proof will be given in several steps in the following sections. Here we make some comments on the  result.

\n
The theorem provides precisely the state of the system we expected for $t>T_{coll}$ in a perturbative form.  In particular, the zero-th order term describes the unperturbed situation, with spin down and the two particles described by the free evolution of the initial state \eqref{psizzeromeno}.  The first order term describes the situation with spin up, particle $1$ moving freely towards the right with a slightly reduced momentum   $P- \eps P^{-1}$ and localized in position around  
$$
x_1(t)=  a P^{-2} \eps + (P- \eps P^{-1})\,  t= a + (P- \eps P^{-1})(t- T_{coll})\, $$ while particle $2$ moves freely towards the left with unperturbed momentum $-P$ and is localized in position around  
$$x_2(t)= -P t= -a - P(t-T_{coll})\,.$$ 
Finally, the rest is shown to be $O(\eps^2)$. 

\n
On the basis of  Theorem \ref{th11} and Born's rule we can compute approximate formulas for the probabilities of  outcomes of  measurements on the system, performed for $t>\tc$. In particular, we compute the probability to find the spin up or down, denoted by $\mathcal P_{u}$, $\mathcal P_{d}$ respectively, and the probability to find momentum $-P$ for particle $2$ and the spin up or down, denoted $\mathcal P_{-,u}$, $\mathcal P_{-,d}$ respectively. 
By \eqref{psii}, \eqref{th2}, \eqref{th3} we find

\begin{align}
&\mathcal P_u = \alpha \, \eps^2 + O (\eps^3) \,, \;\;\;\;\;\;\;\;\;\;\;\;\; \mathcal P_d = 1 + O (\eps)\\
&\mathcal P_{-,u}  = \alpha \, \eps^2 +   O(\eps^3)   \,, \;\;\;\;\;\;\;\; \;\;\mathcal P_{-,d}= \frac{1}{2} + O(\eps)
\end{align}

\n
where 

\be
\alpha = \f{\pi}{P^2} \, |\hat{V} (P^{-1})|^2 .
\ee

\n
This means that, for  $t>\tc$, we have 
\begin{align}
&\frac{\mathcal P_{-,u}}{\mathcal P_u} = 1 + O(\eps) \label{proco1}\,,\\
&\frac{\mathcal P_{-,d}}{\mathcal P_d} = \frac{1}{2} + O(\eps)\,. \label{proco2}
\end{align}

\n
We stress that the above statements hold for any $t>\tc$ and therefore regardless of how large  the distance is between the spin and the position of particle $2$.

\n 
Formula \eqref{proco1} can be considered as a way to express the essence of the EPR argument. 

\n
Indeed, following the line of reasoning of EPR, one should apply the principles of reality and locality (see also Appendix A). Now, Formula \eqref{proco1}  {{shows}} that, if we measure the spin and the result is ``up'', then we can predict that, {almost certainly}, particle $2$ has momentum $-P$.  This means that one should attribute an element of reality associated with the momentum of particle 2 (due to the reality criterion). Moreover, as particle 1 and the spin do not interact with particle 2, such element of reality would be {{existing}} prior to the measurement of the spin (due to the locality principle). On the other hand, before the measurement, Quantum
Mechanics does not provide a representation for such an element of reality and therefore
must be considered as incomplete.
Contrarily to the original EPR argument, in this formulation one does not need to invoke the uncertainty principle.

\n
Conversely, if the spin remains ``down'' after the collision time $\tc$, then  formula \eqref{proco2} says that we cannot predict with certainty the momentum of particle $2$.  

\bigskip
\n
In the framework of non-relativistic
Quantum Mechanics, several models have been proposed to describe a 
process in which the momentum of a particle is revealed by the flip of a spin. In
such a situation, the spin acts as the simplest possible measurement apparatus
\cite{ColemanHepp,Bell,Pascazio,FigariTeta}. Here we apply the same approach to the
mathematical study of the EPR argument. This allows us to treat, as in the
EPR paper but in full mathematical rigour and in a dynamical setting, a continuous variable (the momentum). More specifically, we  define
an initial state of the system and compute  its time evolution perturbatively with a quantitative estimate of the error. We stress that, even though the role of the spin is that
of a measurement apparatus, we always 
consider the whole system ``particle 1 + particle 2 + spin", without having  to resort to the wave packet reduction.

\noindent
As a further feature of our method, we remark that
dealing rigorously with the Schrödinger equation for a continuous system, brings out the necessity of a suitable scaling, in order to isolate the phenomenon we aim at exploring.

\vs
\n
For the convenience of the reader we collect here some further notation that shall
be used in the rest of the paper.


\n
- The Hilbert space of the sub-system made by particle $1$ and the spin:

 \be
 \mathcal K_{1,s} = L^2 (\R) \otimes \C^2.
  \ee
  
  \n
  - The free Hamiltonian and the free unitary group in $\mathcal K_{1,s}$:
 
\begin{equation}
\mathcal H_0^{1,s}= - \frac{\eps^4}{2 } {\mathbb I} \, \Delta_1 
   + \frac{\eps}{2}{\bf \sigma}_3 \, , \;\;\;\;\;\; 
\mathcal U_0^{1,s} (t) = e^ {-i \frac{t}{\eps^2} \mathcal H_0^{1,s}} = 
\left(\!\! \begin{array}{cc}
    U_0(t)\, e^{-i \frac{t}{2\eps}}  & 0  \\
    0 &  U_0(t)\, e^{i \frac{t}{2\eps}}
\end{array}\!\!\right) .
\end{equation}

\n
- the multiplication operator $V^{\eps}$:

\be
V^{\eps}(x_1)= V\Big( \frac{x_1-a}{\eps} \Big) \,, \;\;\;\;\;\; x_1 \in \R.
\ee

\n
- the interacting Hamiltonian and the interacting unitary group in $\mathcal K_{1,s}$:

\begin{equation}
\mathcal H^{1,s}= - \frac{\eps^4}{2 } {\mathbb I} \, \Delta_1 
   + \frac{\eps}{2}{\bf \sigma}_3 +  \eps^2 \, V^{\eps} \, {\bf \sigma}_2\, , \;\;\;\;\;\; \mathcal U^{1,s} (t) = e^ {-i \frac{t}{\eps^2} \mathcal H^{1,s}}.
   \end{equation}

\n
- $(f\otimes g)(x_1,x_2)= f(x_1) g(x_2)$ where $g \in L^2(\R)$ and $f \in L^2(\R)$ or $f \in \mathcal K_{1,s}$.

\vs

\n
- $[ (S \otimes T)(f \otimes g)](x_1,x_2) =  (Sf)(x_1) (Tg)(x_2)$ where $T$ is a linear operator in $L^2(\R)$ and $S$ is a linear operator in $L^2(\R)$ or in $\mathcal K_{1,s}$.

\vs
\n
- The action of the free propagator: $\mathcal U_0(t) $, given explicitly in \eqref{propagatore}.

\vs
\n
- The norms:
\begin{align}
&\|f\|^2 = \int\!\!dx\,  |f(x)|^2 \label{n1}
\end{align}
(however, when necessary, we shall specify the space as $\|f\|^2_{L^2 (\R)}$)

\begin{align}
&\|f\|^2_{L^2_n (\R)} = \int\!\!dx\, |x|^n |f(x)|^2 \label{n2} 
\\
&\| f \|_{W^{1,n}_n (\R) } = \sum_{m=0}^n\int\!\!dx\, (1 + |x|^m) \, |f^{(n-m)}(x)| .
\label{n3} 
\end{align}


\section{Strategy of the proof} \label{strategy}

Here we outline the strategy of the proof. 

\n
By linearity, the evolution of the system is given by

\be
\Psi_t  =  \mathcal U (t) \Psi_0  = N\, \mathcal U(t) \left(\!\! \begin{array}{c} 0  \\ \phi_P\otimes \phi_{-P}   \end{array} \! \! \right) 
+ N\, \mathcal U(t) \left(\!\! \begin{array}{c} 0  \\ \phi_{-P}\otimes \phi_{P}   \end{array} \! \! \right) 
\ee
and, taking into account that   particle $2$ is subject to a free evolution, we can write 
\begin{align}\label{psia}
&\Psi_t = N\, \mathcal U^{1,s} (t) \left(\!\! \begin{array}{c} 0  \\ \phi_P  \end{array} \! \! \right) \otimes  U_0(t) \phi_{-P}   +
N\, \mathcal U^{1,s} (t) \left(\!\! \begin{array}{c} 0  \\ \phi_{-P}  \end{array} \! \! \right)   \otimes  U_0(t) \phi_{P} 
\end{align}

\n
The analysis is then reduced to the study of the evolution of the particle $1$ and the spin. Using Duhamel formula, we have

\begin{equation}\label{duham}
\mathcal U^{1,s} (t)  \left(\!\! \begin{array}{c} 0  \\ \phi_{\pm P}  \end{array} \! \! \right)= 
\mathcal U_0^{1,s} (t)  \left(\!\! \begin{array}{c} 0  \\ \phi_{\pm P}  \end{array} \! \! \right) - i \int_0^t \!\! d\tau \, \mathcal U^{1,s} (t-\tau) V^{\eps} \sigma_2 \,\mathcal U_0^{1,s} (\tau) \left(\!\! \begin{array}{c} 0  \\ \phi_{\pm P}  \end{array} \! \! \right)
\end{equation}

\n
\n
We will show that the last term in \eqref{duham} is negligible 
 when the particle $1$ has a negative momentum (recall that the position $a$ of the spin is assumed to be positive). 
%
Therefore  we write the wave function of the system as

\begin{align}\label{psib}
 \Psi_t  &=   N\, \mathcal U^{1,s} (t) \left(\!\! \begin{array}{c} 0  \\ \phi_P  \end{array} \! \! \right)   \otimes  U_0(t) \phi_{-P}   + N \, e^{i \frac{t}{2 \eps}} \, U_0(t) \phi_{-P}  \otimes U_0(t) \phi_P   \left(\!\! \begin{array}{c} 0  \\1  \end{array} \! \! \right) + N\mathcal L(t)
\end{align}
where
\be\label{Lt}
\mathcal L(t)= - i \int_0^t \!\! d\tau \, \mathcal U^{1,s} (t-\tau) V^{\eps} \sigma_2 \,\mathcal U_0^{1,s} (\tau) \left(\!\! \begin{array}{c} 0  \\ \phi_{- P}  \end{array} \! \! \right)
\otimes U_0(t) \phi_P
\ee

\vs

\n
Let us now consider   the first term in the r.h.s. of \eqref{psib}.  Iterating \eqref{duham}, we find

\begin{align}\label{duham1}
&\mathcal U^{1,s} (t)  \left(\!\! \begin{array}{c} 0  \\ \phi_{ P}  \end{array} \! \! \right)= 
\mathcal U_0^{1,s} (t)  \left(\!\! \begin{array}{c} 0  \\ \phi_{ P}  \end{array} \! \! \right) 
-i \int_0^t \!\!d\tau\, \mathcal U_0^{1,s} (t-\tau)  V^{\eps} \sigma_2\, \mathcal U_0^{1,s} (\tau) \left(\!\! \begin{array}{c} 0  \\ \phi_{ P}  \end{array} \! \! \right)
\nonumber\\
& - \int_0^t \!\! d\tau_1 \, \mathcal U^{1,s} (t-\tau) V^{\eps} \sigma_2\, \mathcal U_0^{1,s} (\tau_1) 
\int_0^{\tau_1} \!\!d\tau_2\, \mathcal U_0^{1,s} (-\tau_2) V^{\eps} \sigma_2\, \mathcal U_0^{1,s} (\tau_2) \left(\!\! \begin{array}{c} 0  \\ \phi_{ P}  \end{array} \! \! \right)\nonumber\\
&= e^{i\frac{t}{2 \eps}} \left(\!\! \begin{array}{c} 0  \\ U_0(t) \phi_{ P}  \end{array} \! \! \right)
- e^{-i \frac{t}{2 \eps}}  \left(\!\! \begin{array}{c}  U_0(t) I(t) \phi_{ P}   \\0 \end{array} \! \! \right) 
+ \mathcal J(t) 
\end{align}
where

\begin{align} 
&I(t)=  \int_0^t \!\!d\tau\, e^{i \frac{ \tau}{\eps}} U_0(-\tau) V^{\eps} U_0(\tau) \label{I(t)0} \\ \label{jt}
&\mathcal J(t)= - \int_0^t \!\! d\tau_1 \, \mathcal U^{1,s} (t-\tau_1) V^{\eps} \sigma_2\, \mathcal U_0^{1,s} (\tau_1) 
\int_0^{\tau_1} \!\!d\tau_2\, \mathcal U_0^{1,s} (-\tau_2) V^{\eps} \sigma_2\, \mathcal U_0^{1,s} (\tau_2) \left(\!\! \begin{array}{c} 0  \\ \phi_{ P}  \end{array} \! \! \right).
\end{align}

\n
Taking into account \eqref{psib} and \eqref{duham1}, the wave function of the system reads

\begin{align}\label{psib1}
&\Psi_t= \mathcal U_0(t) \Psi_0  - N e^{-i \frac{t}{2 \eps}}  \left(\!\! \begin{array}{c}  U_0(t) I(t) \phi_{ P}   \\0 \end{array} \! \! \right) \otimes  U_0(t) \phi_{-P} 
+N  \mathcal J(t) \otimes  U_0(t) \phi_{-P} + N \mathcal L(t)\,.
\end{align}

\n
Let us consider  the operator $I(t)$. Its action   on the function   

\be\label{fxk}
f_{X,K} (x) = \frac{1}{\sqrt{\eps}} \,f\Big( \frac{x}{\eps} -X \Big) \, e^{i \frac{K}{\eps ^2} x} \,, \;\;\;\;X, K >0\,, \;\;\;\; f\in \mathcal S(\R)\,
\ee
can be explicitly computed  and we find (see  Appendix B) 
\begin{align}\label{I(t)}
I(t)f_{X,K} (x) = \frac{e^{i \frac{K}{\eps^2} x} }{\sqrt{2\pi \,\eps} } \int_0^t \!\!d\tau\!\! \int \!\!d\xi\, F(\tau,\xi) \, e^{ \frac{i}{\eps} \Phi(\tau,\xi)}
\end{align}
where
\begin{align} \label{effeffi}
&F(\tau,\xi)= \widehat{V}(\xi) \, f\big( \tau \xi + \frac{x}{\eps} -X \big) \,e^{\frac{i}{2} \tau \xi^2 + \frac{i}{\eps} x \xi }\,,\\
&\Phi(\tau,\xi) = \tau -a \,\xi + K \tau \, \xi\,. 
\end{align}
For notational convenience, we have dropped the dependence on $x, \eps$ of the function $F$. Notice that \eqref{I(t)} has the form of a highly oscillating integral for $\eps \to 0$ and therefore its asymptotic behaviour is characterized by the critical points of the phase

\be
(\tau,\xi)= (\tau_c, \xi_c)= \Big( \frac{a}{K} , -\frac{1}{K} \Big) .
\ee
Notice that $\tau_c$ reduces to $T_{coll}$ when $K = P$.

\n
In order to identify the leading order in the asymptotic expansion for $\eps \to 0$ we write

\be
\Phi(\tau,\xi)= \frac{a}{K} + K (\tau-\tau_c)(\xi-\xi_c)\,,
\ee 
we define the new integration variable $z=\frac{K}{\eps} (\tau-\tau_c)$ and we denote $\Omega_{\eps}= \big[ -\frac{a}{\eps} , \frac{K}{\eps} (t-\tau_c) \big]$. Then we find
\begin{align}\label{xxx}
&I(t)f_{X,K} (x) = \frac{e^{i \frac{K}{\eps^2} x  + \frac{i}{\eps} \frac{a}{K}  } }{\sqrt{2\pi } \,K} \,  \sqrt{\eps} \int_{\Omega_{\eps}} \!\!\!\!dz\,  e^{-i \xi_c z} \!\!\int\!\!d\xi\, e^{i z \xi} F\big(\tau_c + \frac{\eps}{K} z, \xi\big) \nonumber\\
&= \frac{e^{i \frac{K}{\eps^2} x  + \frac{i}{\eps} \frac{a}{K}  } }{\sqrt{2\pi } \,K} \,  \sqrt{\eps} \int\!\!dz\,  e^{-i \xi_c z} \!\!\int\!\!d\xi\, e^{i z \xi} F\big(\tau_c , \xi\big)\nonumber\\
& - \frac{e^{i \frac{K}{\eps^2} x  + \frac{i}{\eps} \frac{a}{K}  } }{\sqrt{2\pi } \,K} \,  \sqrt{\eps} \int_{\R \setminus \Omega_{\eps}} \!\!\!\!dz\,  e^{-i \xi_c z} \!\!\int\!\!d\xi\, e^{i z \xi} F\big(\tau_c , \xi\big)\nonumber\\
&+ \frac{e^{i \frac{K}{\eps^2} x  + \frac{i}{\eps} \frac{a}{K}  } }{\sqrt{2\pi } \,K} \,  \sqrt{\eps} \int_{\Omega_{\eps}} \!\!\!\!dz\,  e^{-i \xi_c z} \!\!\int\!\!d\xi\, e^{i z \xi} \Big(F\big(\tau_c + \frac{\eps}{K} z, \xi\big)-  F\big(\tau_c , \xi\big) \Big)\nonumber\\
&= e^{i \frac{K}{\eps^2} x  + \frac{i}{\eps} \frac{a}{K}  }  \, \frac{\sqrt{2 \pi}}{K} \,  \sqrt{\eps} \, F(\tau_c, \xi_c) + Q_1 (x)+ Q_2(x)\,
\end{align}
where
\begin{align}
&Q_1(t,x)=  - \frac{e^{i \frac{K}{\eps^2} x  + \frac{i}{\eps} \frac{a}{K}  } }{\sqrt{2\pi } \,K} \,  \sqrt{\eps} \int_{\R \setminus \Omega_{\eps}} \!\!\!\!dz\,  e^{-i \xi_c z} \!\!\int\!\!d\xi\, e^{i z \xi} F\big(\tau_c , \xi\big) \\
&Q_2(t,x)=  \frac{e^{i \frac{K}{\eps^2} x  + \frac{i}{\eps} \frac{a}{K}  } }{\sqrt{2\pi } \,K} \,  \sqrt{\eps} \int_{\Omega_{\eps}} \!\!\!\!dz\,  e^{-i \xi_c z} \!\!\int\!\!d\xi\, e^{i z \xi} \Big(F\big(\tau_c + \frac{\eps}{K} z, \xi\big)-  F\big(\tau_c , \xi\big) \Big).
\end{align}

\n
Using  formula \eqref{xxx}, with $X=0$ and $K=P$ we find

\begin{align}\label{ifip}
&I(t) \Phi_P (x)= - \eps \,A(x) + Q_1^0 (t,x) + Q_2^0 (t,x)
\end{align}
where $A$ is defined in \eqref{th2} and $Q_i^0 = Q_i $ for $X=0$ and $K=P$, $i=1,2$. Replacing \eqref{ifip} in \eqref{duham1} we obtain

\begin{align}\label{duham2}
&\mathcal U^{1,s} (t)  \left(\!\! \begin{array}{c} 0  \\ \phi_{ P}  \end{array} \! \! \right)= 
  \left(\!\! \begin{array}{c}\eps  \,e^{-i \frac{t}{2 \eps}} \, U_0(t)A   \\e^{i\frac{t}{2 \eps}}\, U_0(t) \phi_{ P}  \end{array} \! \! \right)
- e^{-i \frac{t}{2 \eps}}  \left(\!\! \begin{array}{c}  U_0(t)(Q_1^{0} (t) +Q_2^{0} (t) )   \\0 \end{array} \! \! \right) 
+ \mathcal J(t) 
\end{align}

\n
In conclusion, using formula \eqref{duham2} in \eqref{psib}, we obtain the following expression for the wave function of the system

\begin{align}\label{psic}
 \Psi_t  &=   N \left(\!\!  \begin{array}{c}\eps \, e^{-i \frac{t}{2 \eps}} \, U_0(t)A   \\e^{i\frac{t}{2 \eps}}\, U_0(t) \phi_{ P}  \end{array} \! \! \right)  \otimes  U_0(t) \phi_{-P}   + N \, e^{i \frac{t}{2 \eps}} \, U_0(t) \phi_{-P}  \otimes U_0(t) \phi_P   \left(\!\! \begin{array}{c} 0  \\1  \end{array} \! \! \right) + N \mathcal R(t)\nonumber\\
 &=\left(\!\! \begin{array}{c}\eps \, N \,e^{-i \frac{t}{2 \eps}} \, U_0(t)A  \otimes U_0(t) \phi_{-P}  \\N\, e^{i\frac{t}{2 \eps}}\, \big[ U_0(t) \phi_{ P} \otimes U_0(t) \phi_{-P} + U_0(t) \phi_{ -P} \otimes U_0(t) \phi_{P} \big]  \end{array} \! \! \right) + N \mathcal R(t)\nonumber\\
 &\nonumber\\
 &\equiv \, \mathcal U_0(t) \Psi_0  + \eps \,  \mathcal U_0(t) \Psi^{(I)}  +N \mathcal R(t)
\end{align}
where $\Psi^{(I)}$ is given in \eqref{psii} and 

\begin{align}\label{rest} 
&\mathcal R(t) =  - e^{-i \frac{t}{2 \eps}}  \left(\!\! \begin{array}{c}  U_0(t)(Q_1^0 (t) +Q_2^0 (t))   \\0 \end{array} \! \! \right) \otimes U_0(t) \phi_{-P} + \mathcal J(t) \otimes U_0(t) \phi_{-P} + \mathcal L(t)
\end{align}

\n
The proof of our result is therefore reduced to the estimate of the norm in $\mathcal K$ of the three terms in the r.h.s. of \eqref{rest}. These estimates will be given in the next sections.

\vs

%
%
%
%

\section{Estimate of $\mathcal R(t)$}\


\subsection{Estimate of $Q_1$ and $Q_2$}

\begin{lemma}
    \label{derivate}
    For the function $F$ defined in \eqref{effeffi} the following identities hold
    \begin{eqnarray}
        \partial^N_\xi F ( \tau, \xi) & = & \sum_{\stackrel{0 \leq k,l,m \leq N} {k+l+m = N}} C_{k,l,m}
  \tau^l \widehat V^{(k)} (\xi)
  f^{(l)} \left( \tau \xi + \frac x \eps - X \right) P_m \left( \tau \xi + \frac x \eps \right) e^{i \Xi (\tau, \xi)} \nonumber \\ \label{derxinf}
  \end{eqnarray}
  and
  \begin{eqnarray} \nonumber
     & & \partial_\varphi \partial^N_\xi F ( \tau_c + \varphi, \xi) \\ & = & \sum_{\stackrel{0 \leq k,l,m \leq N}{k+l+m = N}} 
     C_{k,l,m}
     \left[ l
  (\tau_c + \varphi)^{l -1}\widehat V^{(k)} (\xi)
  f^{(l)} \left(( \tau_c + \varphi) \xi + \frac x \eps - X \right) \nonumber 
   P_m \left( (\tau_c + \varphi) \xi + \frac x \eps \right) \right. \\ &&
  + ( \tau_c + \varphi)^l \xi \widehat V^{(k)}(\xi) 
   f^{(l+1)} \left(( \tau_c + \varphi) \xi + \frac x \eps - X \right) 
   P_m \left( (\tau_c + \varphi) \xi + \frac x \eps \right) \nonumber \\
   & & 
  + ( \tau_c + \varphi)^l \xi \widehat V^{(k)}(\xi) 
   f^{(l)} \left(( \tau_c + \varphi) \xi + \frac x \eps - X \right) 
   P_m' \left( (\tau_c + \varphi) \xi + \frac x \eps \right) \nonumber \\
 &&  \left.
  + \frac i 2 ( \tau_c + \varphi)^l \xi^2 \widehat V^{(k)}(\xi) 
   f^{(l)} \left(( \tau_c + \varphi) \xi + \frac x \eps - X \right) 
   P_m \left( (\tau_c + \varphi) \xi + \frac x \eps \right)  \right] e^{i \Xi (\tau_c + \varphi, \xi)} \nonumber \\ \label{derfi}
  \end{eqnarray}
  where $C_{k,l,m} = \frac{N!}{k! l! m!}$, $P_m$ is a complex polynomial of degree $m$ and $\Xi (\tau, \xi) := \frac 1 2  \tau \xi^2 + \frac 1 \eps x \xi $.
\end{lemma}
\begin{proof}
The result can be proven by direct computation. By applying Leibniz's formula to the expression for $Q_1$ introduced in \eqref{xxx} one gets
\begin{equation*} \begin{split}
      \partial^N_\xi F ( \tau, \xi) \ = \ & \sum_{\stackrel{0 \leq k,l,m \leq N} {k+l+m = N}} C_{k,l,m} (\partial_\xi^k \widehat V (\xi) ) 
      (\partial_\xi^l f \left(
      \tau \xi + \frac x \eps - X \right) )
      (\partial_\xi^m e^{i \Xi (\tau, \xi)}) \\
       = \ & \sum_{\stackrel{0 \leq k,l,m \leq N} {k+l+m = N}} C_{k,l,m}
       \tau^l
       \widehat V^{(k)}(\xi) 
      f^{(l)} \left(
      \tau \xi + \frac x \eps - X \right) 
      (\partial_\xi^m e^{i \Xi (\tau, \xi)}), \
      \end{split}
\end{equation*}
where $C_{k,l,m} = \frac{n!}{k! l! m!},$
as results from iterating Newton's binomial
formula.
Using identity 0.430-2 in \cite{GR} on 
iterate derivatives of composed
functions, and since
$\partial_\xi \Xi(\tau,\xi)=\tau \xi + x / \eps$, $\partial_\xi^2 \Xi (\tau, \xi) = \tau$ and $\partial_\xi^3 \Xi (\tau, \xi) =0$, one has
\begin{equation}
    \partial_\xi^m e^{i \Xi (\tau,\xi)}
    \ = \ \sum_{m_1+2m_2=m} \frac{m!}{m_1! m_2!} i^{m_1+m_2}\frac{\tau^{m_2}}{2^{m_2}} \left( \tau \xi + \frac x \eps
    \right)^{m_1} e^{i \Xi (\tau, \xi)},
\end{equation}
Denoting 
\begin{equation}\label{pm}
 P_m (\zeta) \ := \    \sum_{m_1+2m_2=m} \frac{m!}{m_1! m_2!} i^{m_1+m_2}\frac{\tau^{m_2}}{2^{m_2}} \zeta^{m_1}
\end{equation}
 one obtains \eqref{derxinf}. Identity
\eqref{derfi} is obtained by replacing $\tau$ by $\tau_c + \varphi$
in
\eqref{derxinf} and differentiating in the variable $\varphi$.
\end{proof}
We can now state the main result of the section.
\begin{theorem}
    [Estimate of $Q_1$ and $Q_2$]
    For the quantities $Q_1$ and $Q_2$ defined in \eqref{xxx} the following estimates hold:
    \begin{equation} \begin{split}
    \label{estq1q2}
\| Q_1 (\cdot, t) \| \ \leq & \ C \eps^N, \quad \forall N \in \mathbb N, \\
\| Q_2 (\cdot, t) \| \ \leq & \ C \eps^2 ,
    \end{split}
    \end{equation}
where the constants $C$ are independent of $\eps$.
\end{theorem}

\begin{proof}
By Fubini's theorem, integrating $N$ times by parts in the variable $\xi$ yields
\begin{equation} \label{q1norm}
    \begin{split}
        \| Q_1 (t, \cdot) \|^2 \ = \ & \frac{\eps}{2 \pi K^2} \int_{\R \backslash \Omega_\eps} dz \, \frac{e^{-i \xi_c z}}{z^N}\int_{\R \backslash \Omega_\eps} dz' \, \frac{e^{i \xi_c z'}}{(z')^N} \int d\xi \, e^{iz \xi}
        \int d\xi' \, e^{-iz \xi'} \\ & 
        \int dx \, \partial_\xi^N F (\tau_c, \xi) \overline{\partial_{\xi'}^N
        F (\tau_c, \xi')}
    \end{split}
\end{equation}
We focus on the integral in the variable $x$. By Lemma \ref{derivate} it rewrites as
\begin{eqnarray}
    \nonumber & & \sum_{\stackrel{0 \leq k,l,m \leq N} {k+l+m = N}}
\sum_{\stackrel{0 \leq k',l',m ' \leq N} {k'+l'+m' = N}} C_{k,l,m} C_{k',l',m'}
    \tau_c^{l} \tau_c^{l'} \widehat V^{(k)} (\xi)
    \overline{\widehat V^{(k')} (\xi')} \int dx \,
  f^{(l)} \left( \tau_c \xi + \frac x \eps - X \right) 
  f^{(l')} \left( \tau_c \xi' + \frac x \eps - X \right) \\ && \nonumber 
  P_m \left( \tau_c \xi + \frac x \eps \right) 
  P_{m'} \left( \tau_c \xi' + \frac x \eps \right)
  e^{i \Xi (\tau_c, \xi) - i \Xi (\tau_c,
  \xi')} .
\end{eqnarray}
For the sake of estimating \eqref{q1norm} we pass to the modulus and, with the change of variable $y = \frac x \eps$, we write
\begin{eqnarray*}
    \nonumber & & \left| \sum_{\stackrel{0 \leq k,l,m \leq N} {k+l+m = N}}
\sum_{\stackrel{0 \leq k',l',m ' \leq N} {k'+l'+m' = N}} C_{k,l,m} C_{k',l',m'} \,
    \tau_c^{l} \tau_c^{l'} \widehat V^{(k)} (\xi)
    \overline{\widehat V^{(k')} (\xi')} \right. \\ & & \left.
    \nonumber \int dx \,
  f^{(l)} \left( \tau_c \xi + \frac x \eps - X \right) 
  f^{(l')} \left( \tau_c \xi' + \frac x \eps - X \right) \nonumber 
   P_m \left( \tau_c \xi + \frac x \eps \right) 
  P_{m'} \left( \tau_c \xi' + \frac x \eps \right)
  e^{i \Xi (\tau_c, \xi) - i \Xi (\tau_c,
  \xi')} \right| \\
    \nonumber & \leq & \eps \sum_{\stackrel{0 \leq k,l,m \leq N} {k+l+m = N}}
\sum_{\stackrel{0 \leq k',l',m ' \leq N} {k'+l'+m' = N}} C_{k,l,m} C_{k',l',m'}
    \tau_c^{l} \tau_c^{l'} |\widehat V^{(k)} (\xi)|
    |{\widehat V^{(k')} (\xi')}| \nonumber \\ & & \int dy \,
 | f^{(l)} ( \tau_c \xi + y - X ) |
  |{f^{(l')} ( \tau_c \xi' + y - X )}|
 | P_m ( \tau_c \xi + y ) |
 |{P_{m'} ( \tau_c \xi' + y )} |
  \label{epsextraction}
  \\
  & \leq & 
\eps \sum_{\stackrel{0 \leq k,l,m \leq N} {k+l+m = N}}
\sum_{\stackrel{0 \leq k',l',m ' \leq N} {k'+l'+m' = N}} C_{k,l,m} C_{k',l',m'}
    \tau_c^{l} \tau_c^{l'} 
     \| f^{(l)} ( \cdot - X ) P_m \|
  \|{f^{(l')} ( \cdot- X )} P_{m'} \| 
    |\widehat V^{(k)} (\xi)|
    |{\widehat V^{(k')} (\xi')}| , 
  \end{eqnarray*}
  where we used Cauchy-Schwarz inequality. 
 Then, from \eqref{q1norm} one has
  \begin{equation}
      \begin{split}
          \| Q_1 (t, \cdot) \|^2 \ \leq \ &
          C \eps^2
          \sum_{\stackrel{0 \leq k,l,m \leq N} {k+l+m = N}}
\sum_{\stackrel{0 \leq k',l',m ' \leq N} {k'+l'+m' = N}} C_{k,l,m} C_{k',l',m'}
    \tau_c^{l} \tau_c^{l'} 
     \| f^{(l)} ( \cdot - X ) P_m \|
  \|{f^{(l')} ( \cdot- X )} P_{m'} \| \\
  &
          \int_{\R \backslash \Omega_\eps} \frac{dz}{z^N} \int_{\R \backslash \Omega_\eps} \frac{dz'}{(z')^N} \int d \xi \,
         | \widehat V^{(k)} (\xi) |
         \int d \xi' 
          \, |\widehat V^{(k')} (\xi')| \\
           \leq \ & C \frac {\eps^{2N}} {N-1} \left( \frac{1}{a^{N-1}} + \frac{1}{(k (t - \tau_c))^{N-1}} \right)^2 \sum_{\stackrel{0 \leq k,l,m \leq N} {k+l+m = N}}
\sum_{\stackrel{0 \leq k',l',m ' \leq N} {k'+l'+m' = N}} C_{k,l,m} C_{k',l',m'} \\ &
    \tau_c^{l} \tau_c^{l'}   \| f^{(l)} ( \cdot - X ) P_m \|
  \|{f^{(l')} ( \cdot- X )} P_{m'} \| 
  \| \widehat V^{(k)} \|_1 
   \| \widehat V^{(k')} \|_1 \\
    = \ & C \eps^{2N},
           \end{split}
  \end{equation}
so the first inequality in \eqref{estq1q2} is proven. 

In order to get the second inequality, we write the quantity $Q_2$ in a different way.
First, from \eqref{xxx} and \eqref{derfi},
\begin{equation}
    \begin{split}
        Q_2 (t,x) \ = \ & \frac{e^{i \frac{K}{\eps^2} x  + \frac{i}{\eps} \frac{a}{K}  } }{\sqrt{2\pi } \,K} \,  \sqrt{\eps} \int_{\Omega_{\eps}} \!\!\!\!dz\,  e^{-i \xi_c z} \!\!\int\!\!d\xi\, e^{i z \xi} \Big(F\big(\tau_c + \frac{\eps}{K} z, \xi\big)-  F\big(\tau_c , \xi\big) \Big) \\
        = \ &
    \frac{e^{i \frac{K}{\eps^2} x  + \frac{i}{\eps} \frac{a}{K}  } }{\sqrt{2\pi } \,K} \,  \sqrt{\eps} \int_{\Omega_{\eps}} \!\!\!\!dz\,  \frac{e^{-i \xi_c z}}{1+z^4} \!\!\int\!\!d\xi\, \int_0^{ \frac{z}{K} \eps} d \varphi \,
        e^{i z \xi} \partial_\varphi (1 + \partial_\xi^4) F\big(\tau_c + \varphi, \xi\big) \\
        = \ & \frac{e^{i \frac{K}{\eps^2}x+\frac i \eps a K}}{\sqrt{2\pi}K} \sqrt{\eps} \int_{\Omega_\eps} dz \, \frac{e^{-i \xi_c z}}{1+z^4}
        \int d\xi \,  \int_0^{\frac{z}{K} \eps} d \varphi \,
        e^{i z \xi} \, e^{i \Xi (\tau_c + \varphi, \xi)} \\ &
        \left[ \xi \widehat V (\xi) f' \left((\tau_c + \varphi) \xi + \frac{x}{\eps} - X \right) + \frac{i}{2} \xi^2 \widehat V (\xi) f \left((\tau_c + \varphi) \xi + \frac{x}{\eps} - X \right) 
        \right.
        \\ &
        \sum_{\stackrel{0 \leq k,l,m \leq 4}{k+l+m = 4}} \widetilde C_{k,l,m} \left( l
  (\tau_c + \varphi)^{l -1}\widehat V^{(k)} (\xi)
  f^{(l)} \left(( \tau_c + \varphi) \xi + \frac x \eps - X \right) \nonumber 
   P_m \left( (\tau_c + \varphi) \xi + \frac x \eps \right) \right. \\ &
  + ( \tau_c + \varphi)^l \xi \widehat V^{(k)}(\xi) 
   f^{(l+1)} \left(( \tau_c + \varphi) \xi + \frac x \eps - X \right) 
   P_m \left( (\tau_c + \varphi) \xi + \frac x \eps \right) \nonumber \\
   & 
  + ( \tau_c + \varphi)^l \xi \widehat V^{(k)}(\xi) 
   f^{(l)} \left(( \tau_c + \varphi) \xi + \frac x \eps - X \right) 
   P_m' \left( (\tau_c + \varphi) \xi + \frac x \eps \right) \nonumber \\
 & \left.
  + \frac i 2 ( \tau_c + \varphi)^l \xi^2 \widehat V^{(k)}(\xi) 
   f^{(l)} \left(( \tau_c + \varphi) \xi + \frac x \eps - X \right) 
   P_m \left( (\tau_c + \varphi) \xi + \frac x \eps \right)  \right], \nonumber 
    \end{split}
\end{equation}
where $\widetilde C_{k,l,m} = \frac{24}{k! l! m!}.$
In order to estimate the $L^2-$norm of $Q_2$, we follow the line traced for estimating $Q_1,$ but noticing that the integration in $z$ is on the domain $\Omega_\eps$ instead of its complementary. Furthermore, we will omit the details of the computation since for each term they are 
analogous to those used for estimating $Q_1.$
Then
\begin{equation}
    \begin{split}
        \| Q_2 (t, \cdot) \| \ \leq \ &
        \frac{\eps}{\sqrt{2 \pi} K} \int_{\Omega_{\eps}}   \frac{dz}{1+z^4} 
        \int d\xi \,  \int_0^{\frac{z}{K} \eps} \left| d \varphi \,
         \left(
        | \xi \widehat V (\xi)| \| f' \| + \frac{1}{2} |\xi^2 \widehat V (\xi)| \| f \| 
        \right. \right.
        \\ & +
        \sum_{\stackrel{0 \leq k,l,m \leq 4}{k+l+m = 4}} \widetilde C_{k,l,m} \left( l
  (\tau_c + \varphi)^{l -1} |\widehat V^{(k)} (\xi) | \|
  f^{(l)} \left( \cdot - X \right) \nonumber 
   P_m \| \right. \\ & \left. 
  + ( \tau_c + \varphi)^l |\xi \widehat V^{(k)}(\xi) | \|
   f^{(l+1)} \left( \cdot - X \right) 
    P_m \|
  + ( \tau_c + \varphi)^l | \xi \widehat V^{(k)}(\xi) | \|
   f^{(l)} \left( \cdot - X \right) 
   P_m'  \| \right. \\
 & \left. \left.
  + \frac 1 2 ( \tau_c + \varphi)^l  |\xi^2 \widehat V^{(k)}(\xi) | \|
   f^{(l)} \left( \cdot - X \right) 
   P_m \|   \right) \right|. \nonumber 
            \end{split}
\end{equation}
For each term, the integral in $\varphi$ gives a contribution of order $\eps$, the integral
in $\xi$ is finite since $V$ is in the Schwartz class, and the integral in $z$ becomes
$$ \int_{\Omega_\eps} dz \, \frac{|z|}{1 + z^4} \ \leq \ C.$$
Then we conclude
$$ \| Q_2 (t, \cdot) \| \ \leq \ C \eps^2$$
and the proof is complete.

\end{proof}

\subsection{Estimate of ${\mathcal J} (t)$}

\begin{proposition}
    Let us fix $t > T_{coll}.$ Then
    $$\| {\mathcal J} (t) \|_{\mathcal K} \ \leq \ C \ \eps^2,$$
    where the constant $C$ depends on $t$, $a, P, f, V$.
\end{proposition}

\begin{proof}
    From the definition \eqref{jt} of ${\mathcal J} (t)$
    one has
    \begin{equation}
        \label{anotherduhamel}
        {\mathcal J} (t) \ = \ {\mathcal J}_0 (t) +
        {\mathcal M} (t)
\end{equation}
where
\begin{equation}
    \label{j0}
    {\mathcal J_0} (t) \ = \
   -  \int_0^{t} \!\!d\tau_1 \, \mathcal U_0^{1,s} (t -\tau_1) V^{\eps} \sigma_2\, \mathcal U_0^{1,s} (\tau_1)
    \int_0^{\tau_1} \!\!d\tau_2\, \mathcal U_0^{1,s} (-\tau_2) V^{\eps} \sigma_2\, \mathcal U_0^{1,s} (\tau_2)
    \left(\!\! \begin{array}{c} 0  \\ \phi_{ P}  \end{array} \! \! \right)
\end{equation}
and
\begin{equation} \label{mt} \begin{split}
{\mathcal M} (t) \ = \ & i 
 \int_0^t \!\! d\tau_1 \, \mathcal U^{1,s} (t-\tau_1) V^{\eps} \sigma_2\, \mathcal U_0^{1,s} (\tau_1) 
\int_0^{\tau_1} \!\!d\tau_2\, \mathcal U_0^{1,s} (-\tau_2) V^{\eps} \sigma_2\, \mathcal U_0^{1,s} (\tau_2) \\ &
\int_0^{\tau_2} \!\!d\tau_3\, \mathcal U_0^{1,s} (-\tau_3) V^{\eps} \sigma_2\, \mathcal U_0^{1,s} (\tau_3)
\left(\!\! \begin{array}{c} 0  \\ \phi_{ P}  \end{array} \! \! \right).
\end{split}
        \end{equation}
First we estimate $\| {\mathcal J}_0 (t) \|_{\mathcal K}.$
We preliminarily observe that, since the matrix $\sigma_2$ acts twice, the only non-trivial
component of ${\mathcal J} (t)$ is the second one, so
the $\mathcal K$-norm of $\mathcal J (t)$ equals the norm of such  component, denoted by ${\mathcal J}_0^{(2)}$, in the space $L^2 (\R)$. 

By a straightforward computation one finds
\begin{equation} \label{J02}
\begin{split}
   {\mathcal J}_0^{(2)} (t)   \ = \ &
    - U_0 (t)
 \int_0^{t} \!\!d\tau_1 \, e^{-i \frac{\tau_1}{\eps}} U_0 (-\tau_1) V^{\eps}  U_0 (\tau_1)
    \int_0^{\tau_1} \!\!d\tau_2\, 
    e^{i \frac{\tau_2}{\eps}}
    U_0 (-\tau_2) V^{\eps}  U_0^{1,s} (\tau_2)
    \phi_{ P} .
\end{split}
\end{equation}
For convenience we introduce the notation
\begin{equation}
    \begin{split}
        J_0 (t,x)
        \ = \ 
        \left(
\int_0^{t} \!\!d\tau_1 \, e^{-i \frac{\tau_1}{\eps}} U_0 (-\tau_1) V^{\eps}  U_0 (\tau_1)
    \int_0^{\tau_1} \!\!d\tau_2\, 
    e^{i \frac{\tau_2}{\eps}}
    U_0 (-\tau_2) V^{\eps}  U_0^{1,s} (\tau_2)
    \phi_{ P} \right) (x)
\end{split}
\end{equation}
and notice that
\begin{equation}
    \| {\mathcal J}_0 (t) \|_{\mathcal K} =
    \| J_0 (t) \|_{L^2 (\R).}
\end{equation}
Using the explicit expression for the 
propagators and
proceeding as in Appendix B, one
gets
\begin{equation}
    \begin{split}
{J}_0 (t,x)
\ = \ & 
\int_{\R^6} d \xi_1 
\, d \xi_2 \, dy \, dz \, d \eta \, dz \,
\int_0^t d\tau_1 \int_0^{\tau_1} d \tau_2 \, \frac{e^{i 
{\Upsilon
(x,\xi_1, \xi_2,y,z,\eta, \zeta, \tau_1, \tau_2)/\eps}}}{8 \pi^3 \tau_1 \tau_2 \eps^4}
\widehat V (\xi_1)
\widehat V (\xi_2)
 \phi_P (\zeta)
\end{split}
\end{equation}
where
\begin{equation}
    \begin{split}
        \Upsilon
(x,\xi_1, \xi_2,y,z,\eta, \zeta, \tau_1, \tau_2) = & \frac{\tau_2 - \tau_1} \eps +
\frac{z^2 - x^2 + 2 xy - 2 yz}{2 \tau_1 \eps^2} + \frac{\zeta^2 - z^2 + 2 z \eta - 2 \eta \zeta}{2 \tau_2 \eps^2} + \\ & +\frac{y \xi_1 + \eta \xi_2 - a (\xi_1 + \xi_2)}
{\eps}.
    \end{split}
\end{equation}
Using
\begin{equation}
    \begin{split}
        \int_\R
        d \eta \, e^{i \eta \left(\frac{z -  \zeta}
{\tau_2 \eps^2} + \frac {\xi_2} \eps        \right)} \ = \
2 \pi \tau_2 \eps^2 
\delta (z - \zeta +
\eps \tau_2 \xi_2) \\
\int_\R
        d y \, e^{i y
        \left(\frac{x-z}
{\tau_1 \eps^2} + \frac {\xi_1} \eps        \right)} \ = \
2 \pi \tau_1 \eps^2 
\delta (x-z +
\eps \tau_1 \xi_1) 
\end{split}
\end{equation}
and then integrating
in the variables $z$ and $\zeta,$ one obtains
\begin{equation}
    \begin{split}
        J_0 (t,x) 
        \ = \ &
        \frac {e^{iPx/\eps^2}} {2 \pi \sqrt \eps} \int_{\R^2}
d \xi_1 \, d \xi_2 \int_0^t d \tau_1 \int_0^{\tau_1} d \tau_2 \, \widehat V (\xi_1) \widehat V (\xi_2) \, f \left(
\frac x \eps +  \tau_1 \xi_1 +
\tau_2 \xi_2 \right)
\\
& \times e^{i \Upsilon_1^\eps (x, \xi_1, \xi_2, \tau_1,
\tau_2)}
e^{i \Upsilon_2^\eps (x, \xi_1, \xi_2, \tau_1,
\tau_2)/\eps}
\end{split}
\end{equation}
where
\begin{equation}
\begin{split}
    \label{upsilon}
    \Upsilon_1^\eps (x,\xi_1,\xi_2, \tau_1, \tau_2) & \ = \ \f x \eps (\xi_1 + \xi_2) + \frac 1 2
    (\tau_1 \xi_1^2 + \tau_2 \xi_2^2) + \tau_1 \xi_1 \xi_2 \\
    \Upsilon_2^\eps (x,\xi_1,\xi_2, \tau_1, \tau_2) & \ = \
    {\tau_2- \tau_1} -
     a (\xi_1 + \xi_2) + P (\tau_1 \xi_1 + \tau_2 \xi_2) \\
     & \ = \ P (\tau_1 - \tau_c) (\xi_1 + \xi_c) + P (\tau_2 -
     \tau_c) (\xi_2 - \xi_c),
\end{split}
\end{equation}
where, as in Section \ref{strategy}, we set
$\tau_c = \frac a P$, $\xi_c = - \frac 1 P$,
perform
the changes of variable
$z_i = \frac{P}{\eps} (\tau_i - \tau_c), \ i = 1,2$, and obtain
\begin{equation}
    \begin{split}
        J_0 (t,x) 
        \ = \ &
        \frac {e^{iPx/\eps^2}} {2 \pi P^ 2} \eps^{\frac 3 2}
        \int_{\Omega_{1, \eps}}
        dz_1  \, e^{i z_1 \xi_c}
        \int_{\Omega_{2, \eps}}
        d z_2 \, e^{-i z_2 \xi_c}
        \int_\R d\xi_1 \,
        e^{i z_1 \xi_1}
        \int_\R d\xi_2 \,
        e^{i z_2 \xi_2}  \\ & \times \
        \widehat V (\xi_1) \, \widehat V (\xi_2) \,
        f \left( \frac x \eps
        + \tau (z_1) \xi_1 + \tau (z_2) \xi_2 \right) \,
        e^{i \Lambda_\eps (x/ \eps,z_1, z_2, \xi_1, \xi_2)}
\end{split}
\end{equation}
where we introduced the notation $\tau (z) = \frac{\eps z}{P} + \tau_c,$ $\Omega_1 = [ -a/\eps, P (t - \tau_c)  / \eps]$,  $\Omega_2 = [ -a/\eps, P (\tau(z_1) - \tau_c)  / \eps]$,  and
\begin{equation}
\label{lambda}
\Lambda_\eps (y, z_1, z_2, \xi_1, \xi_2) = \frac 1 2 \tau (z_1) (\xi_1^2 + \xi_1 \xi_2) + y (\xi_1 + \xi_2)  
+ \frac 1 2 \tau(z_2) \xi_2^2.
\end{equation}
    Proceeding analogously to  \eqref{q1norm}, we use the identity
    $$e^{iz\xi} \ = \ \frac{1}{1+z^2} (1 - \partial_\xi^2) e^{iz\xi}$$
for the variables $\xi_1$ and $\xi_2$, and finally obtain 
\begin{equation}
    \begin{split}
        J_0 (t,x) 
        \ = \ &
        \frac {e^{iPx/\eps^2}} {2 \pi P^ 2} \eps^{\frac 3 2}
        \int_{\Omega_{1, \eps}}
        dz_1  \, \frac{e^{i z_1 \xi_c}}{1 + z_1^2}
        \int_{\Omega_{2, \eps}}
        d z_2 \, \frac{e^{-i z_2 \xi_c}}{1 + z_2^2}
        \int_\R d\xi_1 \,
        e^{i z_1 \xi_1}
        \int_\R d\xi_2 \,
        e^{i z_2 \xi_2}  \\ & \times \ (1 - \partial_{\xi_1}^2)
        (1 - \partial_{\xi_2}^2)
        \widehat V (\xi_1) \, \widehat V (\xi_2) \,
        f \left( \frac x \eps
        + \tau(z_1) \xi_1 + \tau(z_2) \xi_2 \right) \,
        e^{i \Lambda_\eps (x,z_1, z_2, \xi_1, \xi_2)} \\
\ = \ & \frac {e^{iPx/\eps^2}} {2 \pi P^ 2} \eps^{\frac 3 2}
        \int_{\Omega_{1, \eps}}
        dz_1  \, \frac{e^{i z_1 \xi_c}}{1 + z_1^2}
        \int_{\Omega_{2, \eps}}
        d z_2 \, \frac{e^{-i z_2 \xi_c}}{1 + z_2^2}
        \int_\R d\xi_1 \,
        e^{i z_1 \xi_1}
        \int_\R d\xi_2 \,
        e^{i z_2 \xi_2}  \\ & \times \ 
        \widehat V (\xi_1) \, \widehat V (\xi_2) \,
        f \left( \frac x \eps
        + \tau(z_1)\xi_1 + \tau (z_2)\xi_2 \right) \,
        e^{i \Lambda_\eps (x,z_1, z_2, \xi_1, \xi_2)} \\
        & - \frac {e^{iPx/\eps^2}} {2 \pi P^ 2} \eps^{\frac 3 2}
        \int_{\Omega_{1, \eps}}
        dz_1  \, \frac{e^{i z_1 \xi_c}}{1 + z_1^2}
        \int_{\Omega_{2, \eps}}
        d z_2 \, \frac{e^{-i z_2 \xi_c}}{1 + z_2^2}
        \int_\R d\xi_1 \,
        e^{i z_1 \xi_1}
        \int_\R d\xi_2 \,
        e^{i z_2 \xi_2}  \\ & \times \partial_{\xi_1}^2
        \widehat V (\xi_1) \, \widehat V (\xi_2) \,
        f \left( \frac x \eps
        + \tau (z_1)\xi_1 + \tau (z_2)\xi_2 \right) \,
        e^{i \Lambda_\eps (x,z_1, z_2, \xi_1, \xi_2)}
\\
        & - \frac {e^{iPx/\eps^2}} {2 \pi P^ 2} \eps^{\frac 3 2}
        \int_{\Omega_{1, \eps}}
        dz_1  \, \frac{e^{i z_1 \xi_c}}{1 + z_1^2}
        \int_{\Omega_{2, \eps}}
        d z_2 \, \frac{e^{-i z_2 \xi_c}}{1 + z_2^2}
        \int_\R d\xi_1 \,
        e^{i z_1 \xi_1}
        \int_\R d\xi_2 \,
        e^{i z_2 \xi_2}  \\ & \times \partial_{\xi_2}^2
        \widehat V (\xi_1) \, \widehat V (\xi_2) \,
        f \left( \frac x \eps
        + \tau (z_1)\xi_1 + \tau (z_2) \xi_2 \right) \,
        e^{i \Lambda_\eps (x,z_1, z_2, \xi_1, \xi_2)}
\\
        & + \frac {e^{iPx/\eps^2}} {2 \pi P^ 2} \eps^{\frac 3 2}
        \int_{\Omega_{1, \eps}}
        dz_1  \, \frac{e^{i z_1 \xi_c}}{1 + z_1^2}
        \int_{\Omega_{2, \eps}}
        d z_2 \, \frac{e^{-i z_2 \xi_c}}{1 + z_2^2}
        \int_\R d\xi_1 \,
        e^{i z_1 \xi_1}
        \int_\R d\xi_2 \,
        e^{i z_2 \xi_2}  \\ & \times \partial_{\xi_1}^2 \partial_{\xi_2}^2
        \widehat V (\xi_1) \, \widehat V (\xi_2) \,
        f \left( \frac x \eps
        + \tau(z_1)\xi_1 + \tau (z_2)\xi_2 \right) \,
        e^{i \Lambda_\eps (x,z_1, z_2, \xi_1, \xi_2)} \\
        = & J_{0,1} (t,x) + J_{0,2} (t,x) + J_{0,3} (t,x) + + J_{0,4} (t,x), 
\end{split}
\end{equation}
For the convenience of the reader, we focus on the estimate of $J_{0,1} (t)$
and $J_{0,4} (t)$. The two remaining terms can be treated in the same way.
\begin{equation}
    \begin{split}
\| J_{0,1} (t, \cdot) \| \ \leq & \ 
\frac {1} {2 \pi P^2} \eps^{\frac 3 2}
        \int_{\Omega_{1, \eps}}
       \frac{d z_1}{1 + z_1^2}
        \int_{\Omega_{2, \eps}}
         \frac{{d z_2}}{1 + z_2^2}
        \int_\R d\xi_1 \,
         \left| \widehat V (\xi_1) \right| 
        \int_\R d\xi_2 \,
       \left| \widehat V (\xi_1) \right| 
        \\ & \times \ 
      \left\| f \left( \frac x \eps
        + \tau (z_1) \xi_1 + \tau (z_2) \xi_2 \right) \right\|_{L^2 (\R_x)}
         \\
         = & \ C \, \eps^2.
    \end{split}
\end{equation}
Proceeding with $J_{0,4}$, one has to
distribute two derivatives  in $\xi_1$ and two in $\xi_2$ among the factors in the 
integrand. Then
\begin{equation} \label{j04}
    \begin{split}
\| J_{0,4} (t, \cdot) \| \ \leq &  \ 
\frac {\eps^{2}} {2 \pi P^2} 
\sum_{\stackrel{0 \leq k_1, \ell_1, m_1 \leq 2} { k_1+ \ell_1+  m_1 = 2}}
\sum_{\stackrel{0 \leq k_2, \ell_2, m_2 \leq 2} { k_2+ \ell_2+  m_2 = 2}}
\int_{\Omega_{1, \eps}}  \, \frac{{d z_1}}{1 + z_1^2}
        \int_{\Omega_{2, \eps}}
         \frac{{d z_2}}{1 + z_2^2} \\ &
        \int_\R d\xi_1 \,
         \left| \widehat V^{(k_1)} (\xi_1) \right| 
        \int_\R d\xi_2 \,
       \left| \widehat V^{(k_2)} (\xi_1) \right| 
        \ 
       \tau(z_1)^{\ell_1} (z_1) \tau(z_2)^{\ell_2} (z_2)\\ &
       \| f^{(\ell_1 + \ell_2)} \left( y 
        + \tau(z_1) \xi_1 + \tau(z_2) \xi_2 \right)
\partial_{\xi_1}^{m_1} \partial_{\xi_2}^{m_2} \Lambda_\eps \left(
 y , z_1, z_2, \xi_1, \xi_2 \right)
        \|_{L^2 (\R_y)}
\end{split}
\end{equation}
By \eqref{lambda}
$$
 \partial_{\xi_1}^{m_1} \partial_{\xi_2}^{m_2} \Lambda_\eps \left(
 y , z_1, z_2, \xi_1, \xi_2 \right) = P (y, \eps z_1,  \eps z_2,
 \xi_1, \xi_2),
$$
where $P$ is a polynomial whose maximum degree is one in   $\eps z_1$,
 $\eps z_2$ and $y$, and 
two in  $\xi_1$ and $\xi_2$.
So, one has
\begin{equation} \begin{split}
\label{smontanorma}
     & \| f^{(\ell_1 + \ell_2)} \left(  y 
        + \tau (z_1) \xi_1 + \tau (z_2) \xi_2 \right)
\partial_{\xi_1}^{m_1} \partial_{\xi_2}^{m_2} \Lambda_\eps \left(
 y , z_1, z_2, \xi_1, \xi_2 \right) \|_{L^2(\R_y)}  \\
  \leq  \ &  C  \, (1 + \eps| z_1| + \eps | z_2| +  \xi_1^2+  \xi_2^2 ) \| (1 + |y|) f^{(\ell_1 + \ell_2)} (y + \tau (z_1) \xi_1 + \tau
        (z_2) \xi_2) \|_{L^2 (\R_y)} \\
         \leq  \ &  C \, (1 + \eps| z_1| + \eps | z_2| +   \xi_1^2 +  \xi_2^2 ) \| (1 + |y| + | \xi_1| + |\xi_2|  + \eps|z_1| + \eps |z_2|) f^{(\ell_1 + \ell_2)} (y ) \|_{L^2 (\R_y)} \\
         \leq  \ &  C (1 +  \eps^2 z_1^2 +  \eps^2 z_2^2 +  | \xi_1|^3 + | \xi_2|^3 ) \| (1 + |y| ) f^{(\ell_1 + \ell_2)} (y ) \|_{L^2 (\R_y)} \\
         \leq  \ &  C \, (1 + \eps^2 | z_1|^2 + \eps^2 | z_2|^2 +  | \xi_1|^3 + | \xi_2|^3 ). 
         \end{split}
         \end{equation}

        \n
Inserting \eqref{smontanorma} in \eqref{j04} one finally gets
\begin{equation} \begin{split}
\label{j04finally}
\| J_{0,4} (t) \| \ \leq \ & C \,
\frac {\eps^{2}} {2 \pi P^2} 
\sum_{\stackrel{0 \leq k_1, \ell_1, m_1 \leq 2} { k_1+ \ell_1+  m_1 = 2}}
\sum_{\stackrel{0 \leq k_2, \ell_2, m_2 \leq 2} { k_2+ \ell_2+  m_2 = 2}}
\int_{\Omega_{1, \eps}}  \, \frac{dz_1}{1 + z_1^2} 
        \int_{\Omega_{2, \eps}} \, \frac{dz_2}{1 + z_2^2} 
  \\ &
        \int_\R d\xi_1 \, 
         \left| \widehat V^{(k_1)} (\xi_1) \right| 
        \int_\R d\xi_2 \,
        \left| \widehat V^{(k_2)} (\xi_2) \right| 
        (1 + \eps^2 z_1^2 +  \eps^2 z_2^2 +  | \xi_1|^3 + | \xi_2|^3) . \end{split}
\end{equation}
The integrals in $\xi_1$ and $\xi_2$ are finite due to the rapid 
decay of $\widehat V.$ The integrals in $z_1$ and $z_2$ prove finite too, owing to the elementary estimate
$$
\int_{\Omega_{\eps}} \frac {(\eps |z|)^\beta} {1 + z^2} \, dz 
\ \leq \  \eps^\beta \frac{P^\beta}{\eps^\beta} t^\beta  \int_\R  \frac {dz} {1 + z^2}  
\ \leq \ C ,
$$
where we used $\sup | z | \leq \frac P \eps t $ in the intergation domain.

Then, from \eqref{j04finally},
\begin{equation} \begin{split}
    \| {\mathcal J}_0 (t) \|_{\mathcal K}  \ = \ &
    \left\|
 \int_0^{t} \!\!d\tau_1 \, e^{-i \frac{\tau_1}{\eps}} U_0 (-\tau_1) V^{\eps}  U_0 (\tau_1)
    \int_0^{\tau_1} \!\!d\tau_2\, 
    e^{i \frac{\tau_2}{\eps}}
    U_0 (-\tau_2) V^{\eps}  U_0^{1,s} (\tau_2)
    \phi_{ P}  \right\|_{L^2 (\R)} \ \leq \ C\eps^2 .
\end{split}
\end{equation}
To estimate $\mathcal M (t)$, from \eqref{mt}, using the unitarity of
$ \mathcal U^{1,s} (t-\tau_1)$,
\begin{equation}  \begin{split}
\| {\mathcal M} (t) \|_\mathcal K \ \leq \ &  \| V \|_\infty
 \int_0^t \!\! d\tau_1 \, \| {\mathcal {J}}_0 (\tau_1) \|_{\mathcal K} \ \leq \ C \eps^2.
\end{split}
        \end{equation}

\end{proof}

\vs

\section{Estimate of $\mathcal L(t)$}

\begin{proposition}
Let us fix $t>0$. Then  
\begin{align} 
 \| \mathcal L(t)\|_{\mathcal K} \leq C_n (t)  \, \eps^n \;\;\;\;\;\; \forall n \in \N 
 \end{align}
where the dependence of the constant $C_n (t)$ on $t$ and on the other parameters $ a, P, f, V$ is given during the proof.
\end{proposition} 

\begin{proof}

From the definition of $\mathcal L(t)$ (see \eqref{Lt}) we have

\begin{align}\label{Lt1}
 \| \mathcal L(t)\|_{\mathcal K} &= \left\| \int_0^t \!\! d\tau \,  \mathcal U^{1,s} (t-\tau) V^{\eps} \sigma_2 \, \mathcal U_0^{1,s} (\tau) \left(\!\! \begin{array}{c} 0  \\\phi_{-P}   \end{array} \! \! \right) \otimes U_0(t) \phi_P 
\right\|_{\mathcal K} \nonumber\\
& \leq \int_0^t \!\!d\tau\left\| V^{\eps} \sigma_2 \, \mathcal U_0^{1,s} (\tau)     \left(\!\! \begin{array}{c} 0  \\\phi_{-P}  \end{array} \! \! \right) \right\|_{\mathcal K^{1,s} }  \|U_0(t) \phi_P \| \nonumber\\
& = \int_0^{t} \!\!d\tau\, \left\| V^{\eps} U_0(\tau) \phi_{-P} \right\| \,.
\end{align}

\noindent
The estimate of the $L^2$-norm of $ V^{\eps}  U_0(\tau) \phi_{-P} $ is elementary  if one takes into account that for $\eps$ small the function $V^{\eps}$ is strongly concentrated around $x=a >0$ while the function $ U_0(\tau) \phi_{-P} $ is strongly concentrated around $x=- P\tau \leq 0$. Then we write

\begin{align} \label{vup}
&\left\| V^{\eps} U_0(\tau) \phi_{-P} \right\|^2 \nonumber\\
&= \!\int_{-\infty}^{a/2} \!\! dx\, V\Big( \frac{x-a}{\eps} \Big)^2 | U_0(\tau) \phi_{-P}(x) |^2 + \int_{a/2}^{+\infty} \!\! dx\, V\Big( \frac{x-a}{\eps} \Big)^2 | U_0(\tau) \phi_{-P}(x) |^2\nonumber\\
& \leq \|U_0(\tau) \phi_{-P} \|_{L^{\infty} (\R) }^2 \int_{-\infty}^{a/2} \!\! dx\, V\Big( \frac{x-a}{\eps} \Big)^2
+ \|V\|_{L^{\infty} (\R) }^2 \int_{a/2}^{+\infty} \!\! dx\,  | U_0(\tau) \phi_{-P}(x) |^2 .
\end{align}

\noindent
Let us recall that 

\begin{align}
U_0(\tau) \phi_{-P} (x) &= \frac{1}{\sqrt{ 2 \pi i \eps \tau}} \int \!\!dy\, e^{i \frac{ (x-y)^2}{ 2 \eps^2 \tau } } \, \phi_{-P} (y) \nonumber\\
&= e^{ -i \frac{ P^2}{ 2 \eps^2 } \tau   -i \frac{ xP}{ \eps^2}  } \, U_0 (\tau) f\! \left( \frac{x +P \tau}{ \eps} \right) \nonumber\\
&=e^{ -i \frac{ P^2}{ 2 \eps^2 } \tau   -i \frac{ xP}{ \eps^2}  } \, \frac{1}{\sqrt{2\pi} } \int\!\!dk\, e^{i \frac{ x+P\tau}{\eps} 
 \,k } \, e^{-ik^2 \tau} \, \hat{f} (k) \,.
\end{align}

\noindent
Hence, from \eqref{vup} we have 

\begin{align} \label{vup1}
&\left\| V^{\eps} U_0(\tau) \phi_{-P} \right\|^2 \nonumber\\
&\leq \frac{ \|\hat{f}\|_{L^1(\R)}^2 }{ 2 \pi  } \int_{-\infty}^{a/2} \!\! dx\, V\Big( \frac{x-a}{\eps} \Big)^2 
+ \|V\|_{L^{\infty} (\R) }^2 \int_{a/2}^{+\infty} \!\! dx\,  \left| U_0(\tau) f \left( \frac{x +P \tau}{ \eps} \right) \right|^2 \nonumber\\
&= \frac{ \|\hat{f}\|_{L^1(\R)}^2 }{ 2 \pi  } \, \eps \!
\int_{-\infty}^{- \eps^{-1} a/2} \!\! dy\, V(y)^2 
+  \|V\|_{L^{\infty} (\R) }^2 \, \eps \!\int_{ \eps^{-1} ( P \tau + a/2) }^{+\infty} \!\! dy\,  \left| U_0(\tau) f(y) \right|^2\,. 
\end{align}

\noindent
It remains to estimate the two integrals in the last line of \eqref{vup1}. For the first one we have

\begin{align}\label{vup2}
\int_{-\infty}^{- \eps^{-1} a/2} \!\! dy\, V(y)^2 &\leq \eps^{ 2n-1 } \left( \frac{ 2 }{ a } \right)^{\! 2n-1} \!\!
\int_{-\infty}^{- \eps^{-1} a/2} \!\! \!dy\, |y|^{2n-1} V(y)^2 \nonumber\\
&\leq 
\eps^{ 2n-1 } \left( \frac{ 2 }{ a } \right)^{\! 2n-1} \!\! \|  V \|^2_{L^2_{2n-1} (\R)} 
\end{align}
for any $n \in \N$, where $\| \cdot \|_{L^2_{n} (\R)}$ denotes the  weighted $L^2$-norm defined in \eqref{n2}. 

\noindent 
For the second   second integral in \eqref{vup1}, a repeated integration by parts  yields

\begin{align}\label{vup3}
&\int_{ \eps^{-1} ( P \tau + a/2) }^{+\infty} \!\! dy\,  \left| U_0(\tau) f(y) \right|^2 = \frac{1}{2 \pi} \int_{ \eps^{-1} ( P \tau + a/2) }^{+\infty} \!\! dy
\left| \int \!\!dk\, e^{iyk} e^{-i k^2 \tau} \hat{f}(k) \right|^2 \nonumber\\
& = \frac{1}{2 \pi} \int_{ \eps^{-1} ( P \tau + a/2) }^{+\infty} \!\! dy\frac{1}{y^{2n} } 
\left| \int \!\!dk\, \left( \frac{d^n}{dk^n}  e^{iyk} \right) e^{-i k^2 \tau} \hat{f}(k) \right|^2 \nonumber\\
& = \frac{1}{2 \pi} \int_{ \eps^{-1} ( P \tau + a/2) }^{+\infty} \!\! dy \, \frac{1}{y^{2n} } 
\left| \int \!\!dk\,   e^{iyk}   \frac{d^n}{dk^n} \left( e^{-i k^2 \tau} \hat{f}(k) \right) \right|^2 \nonumber\\
&\leq \frac{1}{2 \pi} \left( \int \!\!dk\,     \left| \frac{d^n}{dk^n} \left( e^{-i k^2 \tau} \hat{f}(k) \right)  \right| \right)^{\!2} \int_{ \eps^{-1} ( P \tau + a/2) }^{+\infty} \!\! dy\, \frac{1}{y^{2n} }  \nonumber\\
&= \eps^{2n-1} \left( \frac{2}{ a + 2P\tau } \right)^{\! 2n-1}  \!\! \frac{1}{2\pi(2n-1)} 
\left( \int \!\!dk\,     \left| \frac{d^n}{dk^n} \left( e^{-i k^2 \tau} \hat{f}(k) \right)  \right| \right)^{\!2} .
\end{align}

\noindent
Let us notice that

\begin{align}\label{vup4}
& \left| \frac{d^n}{dk^n} \left( e^{-i k^2 \tau} \hat{f}(k) \right)  \right| \leq 
\sum_{m=0}^n 
\left(\!\! \begin{array}{c} n  \\m  \end{array} \! \! \right)
\left| \frac{d^n}{dk^n}  e^{-i k^2 \tau}   \right|
\left| \hat{f}^{(n-m)} (k) \right| \nonumber\\
&= \sum_{m=0}^n 
\left(\!\! \begin{array}{c} n  \\m  \end{array} \! \! \right) 
\tau^{m/2} 
\left|  \left( \frac{ d^m}{dq^m} e^{-iq^2} \right)_{\!\! q=k \sqrt{\tau} }\right|
\left| \hat{f}^{(n-m)} (k) \right| \nonumber\\
&\leq \sum_{m=0}^n 
\left(\!\! \begin{array}{c} n  \\m  \end{array} \! \! \right) 
\tau^{m/2} 
c_m ( 1 + \tau^{m/2} |k|^m ) 
\left| \hat{f}^{(n-m)} (k) \right| \nonumber\\
&\leq \sum_{m=0}^n 
\left(\!\! \begin{array}{c} n  \\m  \end{array} \! \! \right) 
c_m (1 + \tau^m) ( 1 +  |k|^m ) 
\left| \hat{f}^{(n-m)} (k) \right| \nonumber\\
&\leq c_n (1 + \tau^n) \sum_{m=0}^n \left( 1  + |k|^m \right) 
\left| \hat{f}^{(n-m)} (k) \right|
\end{align}
where $c_l$ denotes a numerical constant depending on $l \in \N$. Using \eqref{vup3} in \eqref{vup4}, we find

\begin{align}\label{vup5}
&\int_{ \eps^{-1} ( P \tau + a/2) }^{+\infty} \!\! dy\,  \left| U_0(\tau) f(y) \right|^2 \leq
\eps^{2n-1} \frac{ (1 + \tau^n )^2}{( a + 2 P \tau )^{2n-1} } \, c_n \, \|\hat{f} \|^2_{W^{1,n}_n (\R) }
\end{align}
where $\| \cdot \|_{W^{1,n}_n (\R) }$ denotes the weighted Sobolev norm defined in \eqref{n3}. 

\noindent
Taking into account \eqref{Lt1}, \eqref{vup1}, \eqref{vup2}, \eqref{vup5}  we obtain

\begin{align}
&\| \mathcal L (t) \|_{\mathcal K} \leq 
\eps^n \, c_n \left( 
t \, a^{-n + 1/2} \|\hat{f}\|_{L^1 (\R)} \| V\|_{L^2_{2n-1} (\R) } 
+ \gamma  (t) \, \|V\|_{L^{\infty} (\R) } \, \|\hat{f} \|_{W^{1,n}_n (\R)} \right) 
\end{align}
where 
\begin{align}
&\gamma(t) = \int_0^t \!\!d \tau\, \frac{ 1 + \tau^n}{(a + 2 P \tau)^{n-1/2} }
\end{align}
is a continuous function satisfying
\begin{align}
&\lim_{t \to 0} \frac{\gamma(t)}{t} = a^{-n + 1/2} \,, \;\;\;\;\;\; 
\lim_{t \to \infty } \frac{\gamma(t)}{ t^{3/2} } = \frac{2^{3/2 -n} }{3} P^{-n + 1/2}\,. 
\end{align}

\noindent
This concludes the proof of the proposition.

\end{proof} 

\vs


\section*{Appendix A}
Here, we briefly recall the original EPR argument \cite{EPR}.
The argument relies  on the following two fundamental assumptions:

\n
- reality criterion (RC):  if, without in any way disturbing a system, we can predict with certainty the value of a physical quantity, then there exists an element of reality corresponding to the physical quantity;

\n
- locality principle (LP): the elements of physical reality of a system localized in a given place cannot be instantaneously influenced by a physical process involving a system, at any distance, that does not interact with the first one.

\n
It is worth noticing that, while RC is explicitly stated by the authors, LP is not, since it is considered  as obvious.
\n
Under these assumptions, the goal of the argument is to prove that Quantum Mechanics is not a complete theory, i.e., there are elements of  reality without a counterpart in the theory. 
\n
They consider a simple model made of two quantum particles in dimension one (set $\hbar=1$ for simplicity) and introduce the corresponding position and momentum observables

$$
P^{(j)} f(x_j) =  \f{1}{i} \f{\partial}{\partial x_j} f(x_j) \,, \;\;\;\;\;\;  X^{(j)} f(x_j) = x_j f(x_j)\,, \;\;\;\;  
\;\;\;\;\;\; j=1,2.
$$


\n
Neglecting some mathematical inconsistencies, they assume that the system is described at some instant of time by the entangled (i.e., non factorised) state

$$
\Psi (x_1,x_2)= \int\!\! dp\, u_p(x_1) \psi_{-p}(x_2) =: \int\!\! dp\, e^{i p x_1} \, e^{-ip(x_2 - x_0)} 
$$

\n
where $x_0$ is an arbitrary fixed point. Notice that $
u_p(x_1)= e^{i p x_1} $  is \say{eigenvector} of $P^{(1)}$ with \say{eigenvalue} $p$ and $ \psi_{-p}(x_2) = e^{-i p (x_2-x_0)}$ is  \say{eigenvector} of $P^{(2)}$ with \say{eigenvalue} $-p$. Moreover, using the identity $\delta(x)= (2 \pi)^{-1} \int \!dp\, e^{-ipx}$, we have 
$$
\Psi(x_1,x_2)= 2 \pi \,\delta(x_2-x_1-x_0) \;\;\;\;\;\; \text{and}\;\;\;\;\;\; \widehat{\Psi} (p_1,p_2) = 2 \pi \,\delta(p_1+p_2).
$$
Hence the state $\Psi$ describes the physical situation where the two particles are at (a possibly large) distance $x_0$, with opposite momentum. Furthermore, the state $\Psi$ can also be represented in the form 
$$
\Psi(x_1,x_2)= 2\pi \!\! \int\!\! dx\, v_x(x_1) \varphi_{x+x_0} (x_2)=: 2 \pi \!\! \int\!\!dx\, \delta (x_1-x) \, \delta (x_2 -x -x_0)
$$
where $  v_x(x_1)= \delta (x_1-x)$ is \say{eigenvector} of $X^{(1)}$ with \say{eigenvalue} $x$ and $ \varphi_{x+x_0} (x_2) = \delta (x_2-x-x_0)$ is \say{eigenvector} of $X^{(2)}$ with \say{eigenvalue} $x+x_0$.

\n
Now the proof proceeds as follows. 
Let us perform a measurement of the momentum $P^{(1)}$ of the first particle and let $\bar{p}$ be the result. By the wave packet collapse, after the measurement the state of the system is 
$$
u_{\bar{p}}(x_1) \psi_{-\bar{p}} (x_2) 
$$
Hence, after the measurement on the first particle,  we can predict with certainty that the second particle has momentum $-\bar{p}$. By RC this means that there is an element of reality corresponding to the observable $P^{(2)}$ of the second particle. By LC this element of reality cannot have been created by the measurement on the first particle and therefore it existed from the beginning. 

\n
Proceeding in a similar way, starting from the same state $\Psi$ one measures  the position of the first particle and let $\bar{x}$ be the result. By the wave packet collapse, after the measurement the state of the system is 

$$
v_{\bar{x} }(x_1) \varphi_{\bar{x}+x_0} (x_2)
$$

\n
Hence, after the measurement on the first particle,  we can predict with certainty that the second particle has position $\bar{x} +x_0$ and so there exists an element of reality corresponding to the position of the second particle.

\n
Thus, given the state $\Psi$ of the system  and performing measurements on the first particle, without in any way disturbing the second particle one can attribute elements of reality corresponding to both momentum and position of the second particle. 

\n
On the other hand, in Quantum Mechanics the uncertainty principle implies that position and momentum cannot be both predicted with certainty. 
Therefore, there are elements of reality that do not have counterparts in the theory. 
It follows that Quantum Mechanics is not complete.

\vs

\section*{Appendix B}

Here we derive formula \eqref{I(t)}. From \eqref{I(t)0} and \eqref{fxk} we have

\begin{align}
&\big( I(t) f_{X,K} \big)(x)= \frac{1}{ 2 \pi \, \eps^{5/2} } \int_0^{t} \!\! d\tau\, \frac{e^{i \frac{\tau}{\eps}}}{\tau} 
\int\!\!dy\, e^{- i \frac{(x-y)^2}{2 \eps^2 \tau}} V\!\left( \!\frac{y-a}{\eps} \!\right) 
\int\!\!dz\, e^{i \frac{(y-z)^2}{2 \eps^2 \tau}}  f\! \left( \frac{z}{\eps} -X\right)  e^{i \frac{K}{\eps^2} z} \nonumber
\end{align}
We write
$$
V\!\left( \!\frac{y-a}{\eps} \!\right) = \frac{1}{\sqrt{2 \pi}} \int\!\! d\xi\, \hat{V} (\xi) \, e^{i \frac{y-a}{\eps} \xi} 
$$
and then

\begin{align}
&\big( I(t) f_{X,K} \big)(x)
= \frac{1}{ (2 \pi)^{3/2}  \, \eps^{5/2} } \int_0^{t} \!\! d\tau\, \frac{e^{i \frac{\tau}{\eps}}}{\tau} \!
\int\!\!dy\, e^{- i \frac{(x-y)^2}{2 \eps^2 \tau}} \!
\int\!\! d\xi\, \hat{V} (\xi) \, e^{i \frac{y-a}{\eps} \xi} \!
\int\!\!dz\, e^{i \frac{(y-z)^2}{2 \eps^2 \tau}}  f\! \left( \frac{z}{\eps} -X\right)  e^{i \frac{K}{\eps^2} z} \nonumber\\
&= \frac{e^{-i \frac{x^2}{2 \eps^2 \tau} } }{ (2 \pi)^{3/2}  \, \eps^{5/2} } \int_0^{t} \!\! d\tau\, \frac{e^{i \frac{\tau}{\eps}}}{\tau} \!
\int\!\!dy\, e^{i \frac{xy}{\eps^2 \tau}} 
\!
\int\!\! d\xi\, \hat{V} (\xi) \, e^{i \frac{y-a}{\eps} \xi} \!
\int\!\!dz\,  f\! \left( \frac{z}{\eps} -X\right)  e^{ i \frac{z^2}{2 \eps^2 \tau} - i \frac{yz}{\eps^2 \tau} +   i \frac{K}{\eps^2} z} \nonumber\\
&= \frac{e^{-i \frac{x^2}{2 \eps^2 \tau} } }{ (2 \pi)^{3/2}  \, \sqrt{\eps} } 
\int_0^{t} \!\! d\tau\, e^{i \frac{\tau}{\eps}} \int\!\!dv\, e^{i x v} 
\int\!\! d\xi\, \hat{V} (\xi) \, e^{- i \frac{a}{\eps} \xi    + i \eps \tau v \xi} \!
\int\!\!dz\,  f\! \left( \frac{z}{\eps} -X\right)  e^{ i \frac{z^2}{2 \eps^2 \tau} - ivz +   i \frac{K}{\eps^2} z} \nonumber
\end{align}
Now we exchange the order of integration and by a formal delta-function computation we find
\begin{align}
&\big( I(t) f_{X,K} \big)(x)
= \frac{e^{-i \frac{x^2}{2 \eps^2 \tau} } }{ (2 \pi)^{3/2}  \, \sqrt{\eps} } 
\int_0^{t} \!\! d\tau\, e^{i \frac{\tau}{\eps}}\!\!
\int\!\! d\xi\, \hat{V} (\xi) \, e^{- i \frac{a}{\eps} \xi   } \!\!
\int\!\!dz\,  f\! \left( \frac{z}{\eps} -X\right)  e^{ i \frac{z^2}{2 \eps^2 \tau} +   i \frac{K}{\eps^2} z} \!\!
\int\!\!dv\, e^{-i v \big( z-x-\eps \tau \xi \big)}  \nonumber\\
&= \frac{e^{-i \frac{x^2}{2 \eps^2 \tau} } }{  \sqrt{2 \pi \eps} } 
\int_0^{t} \!\! d\tau\, e^{i \frac{\tau}{\eps}}\!\!
\int\!\! d\xi\, \hat{V} (\xi) \, e^{- i \frac{a}{\eps} \xi   } \!\!
\int\!\!dz\,  f\! \left( \frac{z}{\eps} -X\right)  e^{ i \frac{z^2}{2 \eps^2 \tau} +   i \frac{K}{\eps^2} z} \,
\delta \big( z-x-\eps \tau \xi \big) \nonumber\\
&= \frac{e^{i \frac{K}{ \eps^2 } x } }{  \sqrt{2 \pi \eps} } 
\int_0^{t} \!\! d\tau\!\!
\int\!\! d\xi\, \hat{V} (\xi) \,  f\! \left( \frac{x}{\eps} + \tau \xi -X\right) 
e^{i \left( \frac{\tau \xi^2}{2} + \frac{x}{\eps} \xi \right)} \,
e^{\frac{i}{\eps} \left( \tau - a \xi + K \tau \xi \right) } \nonumber
\end{align}
which coincides with \eqref{I(t)}. Notice that the above formal computation can  easily be made rigorous by a standard regularization procedure of the integral in the variable $v$.

\vs

\n
{\bf Acknowledgements.} 
This work has been partially financed by European Union - Next Generation EU, Projects MUR-PRIN 2022, reference n. 20225ATSP, 2022TMW2PY, 2022CHELC7.
L.B. and A.T.   acknowledge the support of the GNFM  - INdAM., and R.A. acknowledges the support of the GNAMPA  - INdAM project \say{Analisi spettrale, armonica e stocastica in presenza di potenziali magnetici} in the framework \say{Progetti di ricerca 2025}.

\end{document}